\documentclass[final,3p,times]{elsarticle}

\usepackage{amsmath,amssymb,amsfonts,amsthm}
\usepackage{color}
\usepackage{enumerate}
\usepackage{algorithmic,algorithm}
\usepackage{graphics, graphicx}
\usepackage{float}
\usepackage{relsize}
\usepackage{caption}
\usepackage{subfigure}

\newtheorem{thm}{Theorem}[section]
\newtheorem{lemma}[thm]{Lemma}
\newtheorem{defi}[thm]{Definition}
\newtheorem{rem}{Remark}

\newcommand*{\bigtimes}{\mathop{\hbox{\Large{$\times$}}}}
\newcommand{\coord}{\Upsilon}
\newcommand{\inc}{\textrm{inc}}
\newcommand{\lex}{\textrm{lex}}
\newcommand{\aux}{\textrm{aux}}
\newcommand{\eqcl}{\mathrel{\mathsmaller{\mathsmaller{^{\boldsymbol{\sqsubseteq}}}}}}
\newcommand{\PFDalg}{\texttt{PFD\_of\_Di-Hypergraphs}}

\begin{document}

\sloppy

\begin{frontmatter}


\title{Fast Factorization of {C}artesian products of Hypergraphs}	

\author[GW,SB]{Marc Hellmuth\corref{cor1}}
				\ead{mhellmuth@mailbox.org}
\author[HH]{Florian Lehner}
				\ead{mail@florian-lehner.net}

\address[GW]{Department of Mathematics and Computer Science, University of Greifswald, 
				 Walther-Rathenau-Stra{\ss}e 47, 17487 Greifswald, Germany }
\address[SB]{Center for Bioinformatics, Saarland University, Building E
             2.1, Room 413, P.O. Box 15 11 50, D-66041 Saarbr\"{u}cken,
             Germany   					}
\address[HH]{Department of Mathematics, University of Hamburg, Bundesstra{\ss}e 55, 20146 Hamburg. Germany}

\begin{abstract}
Cartesian products of graphs and hypergraphs  have been studied since the 1960s.
For (un)directed hypergraphs, unique \emph{prime factor decomposition (PFD)} results with
respect to the Cartesian product are known. However, there is still a lack of
algorithms, that compute the PFD of directed hypergraphs with respect to the
Cartesian product. 

In this contribution, we focus on the algorithmic aspects for determining the 
Cartesian prime factors of a finite, connected, directed hypergraph and
present a first polynomial time algorithm to compute its PFD. 
In particular, the algorithm has time complexity $O(|E||V|r^2)$ 
for hypergraphs $H=(V,E)$, where the rank $r$ is the maximum number of vertices 
contained in an hyperedge of $H$. If $r$ is bounded, then this algorithm performs even
in $O(|E|\log^2(|V|))$ time. Thus, our method additionally improves also the 
time complexity of PFD-algorithms designed for undirected hypergraphs
that have time complexity $O(|E||V|r^6\Delta^6)$, where $\Delta$ is 
the maximum number of hyperedges a vertex is contained in.
\end{abstract}

\end{frontmatter}

\sloppy

\section{Introduction}

Products are a common way in mathematics of constructing larger objects from
smaller building blocks. For graphs, hypergraphs, and related set systems
several types of products have been investigated, see
\cite{HOS-12,Hammack:2011a} for recent overviews. 

In this contribution we will focus on the \emph{Cartesian product} of
\emph{directed hypergraphs} that are the common generalization of both directed
graphs and (undirected) hypergraphs. In particular, we present a fast and
conceptually very simple algorithm to find the decomposition of directed
hypergraphs into \emph{prime} hypergraphs (its so-called prime factors), where a
(hyper)graph is called prime if it cannot be presented as the product of two
nontrivial (hyper)graphs, that is, as the product of two (hyper)graphs with at
least two vertices. 

\paragraph{Graphs and the Cartesian Product}
A graph is a tuple $G=(V,E)$ with non-empty set of vertices $V$ and a set of
edges $E$ containing two-element subsets of $V$. If the edges are ordered pairs,
then $G$ is called directed and undirected, otherwise. The Cartesian
\emph{graph} product was introduced by Gert Sabidussi \cite{Sabidussi:60}. As
noted by Szamko{\l}owicz \cite{Szamko:62} also Shapiro introduced a notion of
Cartesian products of graphs in \cite{Shapiro:53}. Sabidussi and independently
V.G. Vizing \cite{Vizing63:CartProd} showed that connected \emph{undirected}
graphs have a representation as the Cartesian product of prime graphs
that is unique up to the order and isomorphisms of the factors. The question
whether one can find the prime factorization of connected undirected graphs in
polynomial time was answered about two decades later by Feigenbaum et al.\
\cite{FHS:85} who presented an $O(|V|^{4.5})$ time algorithm. From then on, a
couple of factorization algorithms for undirected graphs have been developed
\cite{AHI:92,Feder:92,FHS:85,Imrich07:linear,Winkler:87}. The fastest one is due
to Imrich and Peterin \cite{Imrich07:linear} and runs in linear-time
$O(|V|+|E|)$.

For connected \emph{directed} graphs, Feigenbaum showed that the Cartesian
product satisfies the unique prime factorization property \cite{Fei86}.
Additionally, she provided a polynomial-time algorithm to determine the prime
factors which was improved to a linear time approach by Crespelle et al.\
\cite{CTL:13}. 

\paragraph{Hypergraphs and the Cartesian Product}
Hypergraphs are a natural generalization of graphs in which ``edges'' may
consist of more than two vertices. More precisely, a hypergraph is a tuple
$H=(V,E)$ with non-empty set of vertices $V$ and a set of hyperedges $E$, where
each $e\in E$ is an ordered pair of non-empty sets of vertices $e=(t(e),h(e))$.
If $t(e)=h(e)$ for all $e\in E$ the hypergraph is called undirected and
directed, otherwise. Products of hypergraphs have been investigated by several
authors since the 1960s
\cite{Berge:Hypergraphs,Black15,Bretto06:Helly,Bretto:13b,Bretto09:HyperCartProd,Bretto:13,Doerfler79:CoversDirectProd,HNO:14,
Imrich67:Mengensysteme,Imrich70:SchwachKartProd,KA:12,KA:15,OstHellmStad11:CartProd,
Sonntag89:HamCart,Zhu:92}. It was shown by Imrich \cite{Imrich67:Mengensysteme}
that connected \emph{undirected} hypergraphs have a unique prime factor
decomposition (PFD) w.r.t.\ to the Cartesian product, up to isomorphism and the
order of the factors. A first polynomial-time factorization algorithm 
for undirected hypergraphs was proposed by Bretto et al.\ \cite{Bretto:13}.

Unique prime factorization properties for \emph{directed} hypergraphs
were derived by Ostermeier et al.\ \cite{OstHellmStad11:CartProd}.
However, up to our knowledge, no
algorithm to determine the Cartesian prime factors of a connected directed
hypergraph is established, so-far.

\paragraph{Summary of the Results}
In this contribution, we present an algorithm to compute the PFD of connected
directed hypergraphs in $O(|V||E|r^2)$ time, where the rank $r$ denotes the maximum number of vertices
contained in the hyperedges. In addition, if we assume to have
hypergraphs with bounded rank the algorithm runs in $O(|E|\log^2(|V|))$ time. 
Note, as directed hypergraphs are a natural
generalization of undirected hypergraphs, our method
generalizes and significantly improves the time-complexity of the method by Bretto et al.\ \cite{Bretto:13}.
In fact, the algorithm of Bretto et al.\
has time complexity $O(|V||E|\Delta^6r^6)$, where $\Delta$ is the maximum number of hyperedges a vertex is
contained in. Assuming that given hypergraphs have bounded rank $r$ and bounded maximum degree $\Delta$
this algorithm runs therefore in $O(|V||E|)$ time.

We shortly outline our method. Given an arbitrary connected directed hypergraph
$H=(V,E)$ we first compute its so-called 2-section $[H]_2$, that is, roughly
spoken the underlying \emph{undirected} graph of $H$. This allows us to use the
algorithm of Imrich and Peterin \cite{Imrich07:linear} in order to compute the
PFD of $[H]_2$ w.r.t.\ the Cartesian \emph{graph} product. As we will show, this
provides enough information to compute the Cartesian prime factors of the
directed hypergraph $H$. In distinction from the method of Bretto et
al.\ our algorithm is in a sense conceptually simpler, as \emph{(1)} we do not
need the transformation of the hypergraph $H$ into its so-called L2-section and
back, where the L2-section is is an edge-labeled version of the 2-section
$[H]_2$, and \emph{(2)} the test which (collections) of the factors of the
2-section are prime factors of $H$ follows a complete new idea based on
increments of fixed vertex-coordinate positions,
that allows an easy and efficient check to determine the PFD of $H$.

\section{Preliminaries}

\subsection{Basic Definitions}

A \emph{directed hypergraph} $H=(V,E)$ consists of a finite \emph{vertex} set
$V(H):=V$ and a set of \emph{directed hyperedges} or \emph{(hyper)arcs}
$E(H):=E$. Each arc $e\in E$ is an ordered pair of non-empty sets of vertices
$e=(t(e),h(e))$. The sets $t(e)\subseteq V$ and $h(e)\subseteq V$ are called the
\emph{tail} and \emph{head} of $e$, respectively. The set of vertices, that are
contained in an arc will be denoted by $V(e):=t(e)\cup h(e)$. If $t(e)=h(e)$
holds for all $e\in E$, we identify $e$ with $V(e)$, and we call $H=(V, E)$ an
\emph{undirected hypergraph}. 
An undirected hypergraph is an
\emph{undirected graph} if $|V(e))|=2$ for all $e\in E$.
The elements of $E$ are called simply edges, if we consider an undirected graph. 
The hypergraph with $|V|=1$ and $E=\emptyset$ is denoted by $K_1$ and is 
called \emph{trivial}. 

Throughout this contribution, we only consider hypergraphs without multiple
hyperedges and thus, being $E$ a usual set, and without loops, that is, $|V(e)|>1$
holds for all $e\in E$. However, we allow to have hyperedges being properly
contained in other ones, i.e., we might have arcs $e,f\in E$ with
$t(e)\subseteq t(f)$ and $h(e)\subseteq h(f)$.

A \emph{partial hypergraph} or \emph{sub-hypergraph} $H'=(V',E')$ of a
hypergraph $H=(V,E)$, denoted by $H'\subseteq H$, is a hypergraph such that
$V'\subseteq V$ and $E'\subseteq E$. The partial hypergraph $H'=(V',E')$ is
\emph{induced (by $V'$)} if $E' = \{e\in E\mid V(e)\subseteq V'\}$. Induced
hypergraphs will be denoted by $\left\langle V'\right\rangle$.

A \emph{weak path $P$ (joining the vertices $v_0,v_k\in V$)} in a hypergraph
$H=(V,E)$ is a sequence $P=(v_0,e_1,v_1,e_2,\ldots,e_k,v_k)$ of distinct
vertices and arcs of $H$, such that $v_0\in V(e_1)$, $v_k\in V(e_k)$ and $v_j\in
V(e_j)\cap V(e_{j+1})$. A hypergraph $H$ is said to be \emph{weakly connected}
or simply \emph{connected} for short, if any two vertices of $H$ can be joined
by a weak path. A \emph{connected component} of a hypergraph $H$ is a connected
sub-hypergraph $H'\subseteq H$ that is maximal w.r.t.\ inclusion, i.e., there is
no other connected sub-hypergraph $H''\subseteq H$ with $H'\subsetneq H''$.
Usually, we identify connected components $H'=(V',E')$ of $H$ simply by their
vertex set $V'$, since $\left\langle V'\right\rangle = H'$.

A \emph{homomorphism} from $H_1=(V_1, E_1)$ into $H_2=(V_2, E_2)$ is a mapping
$\phi: V_1\rightarrow V_2$ such that $\phi(e)$ is an arc in $H_2$ whenever $e$
is an arc in $H_1$ with the property that $\phi(t(e))=t(\phi(e))$ and
$\phi(h(e))=h(\phi(e))$. A bijective homomorphism $\phi$ whose inverse function
is also a homomorphism is called an \emph{isomorphism}.

The \emph{rank} of a hypergraph $H=(V,E)$ is $r(H)=\max_{e\in E}|V(e)|$.

The \emph{$2$-section $[H]_2$} of a (directed) hypergraph $H=(V,E)$ is the undirected
graph $(V,E')$ with $E'=\left\{xy:=\{x,y\}\subseteq V\mid x\neq y,\,\exists\;
e\in E: x,y\subseteq V(e)\right\}$. In other words, two vertices are linked by an edge
in $[H]_2$ if they belong to the same hyperarc in $H$. Thus, every arc $e\in E$
of $H$ is a \emph{complete graph} in $[H]_2$, i.e., all pairwise different
vertices in $V(e)$ are linked by an edge in $[H]_2$. Complete graphs defined on a vertex
set $V$ will be denoted by $K_{|V|}$. 

We will also deal with equivalence relations, for which the following notations
are needed. For an equivalence relations $R$ we write $\varrho \eqcl R$ to
indicate that $\varrho$ is an equivalence class of $R$. A relation $Q$ is finer
than a relation $R$ while the relation $R$ is coarser than $Q$ if $(e,f)\in Q$
implies $(e,f)\in R$, i.e, $Q\subseteq R$. In other words, for each class
$\varrho$ of $R$ there is a collection $\{ \chi | \chi\subseteq \varrho\}$ of
$Q$-classes, whose union equals $\varrho$. Equivalently, for all $\varphi\eqcl
Q$ and $\psi\eqcl R$ we have either $\varphi\subseteq \psi$ or
$\varphi\cap\psi=\emptyset$.

\begin{rem}
If not stated differently, we assume that the hypergraphs considered in this
contribution are connected. 
\end{rem}

\subsection{The Cartesian Product, (Pre-)Coordinates and (Pre-)Layers}

Let $H_1$ and $H_2$ be two hypergraphs. The \emph{Cartesian product} $H=H_1\Box
H_2$ has vertex set $V(H)=V(H_1)\times V(H_2)$, that is the Cartesian set
product of the vertex sets of the factors and the arc set
			\begin{align*} 
				E(H)=&\big\{ ( \{x\}\times t(f),\{x\}\times h(f)) \mid x\in V(H_1), f\in  E (H_2)\big\} \cup\\
					  & \big\{ (t(e)\times \{y\}\ ,\ h(e)\times \{y\}) \mid e\in  E (H_1), y\in V(H_2)\big\}.
			\end{align*}

Thus, the tuple 
$( \{(x_i,y_i), i\in I\} , \{(x'_j,y'_j), j\in J\}) \text{\ with\ }  x_i,x'_j\in V(H_1),\  y_i,y'_j\in V(H_2),\ i\in I, j\in J$ 
is an arc in $ E (H_1\Box H_2)$ if and only if either
			\begin{align*} 
				\text{\ (i)\ }&   ( \{x_i, i\in I\} , \{x'_j, j\in J\}) \text{\ is an arc in\ }  E(H_1) \text{\ and\ } y_i=y'_j \text{\ for all\ } i\in I, j\in J \text{\ or\ }\\
				\text{\ (ii)\ }&  ( \{y_i, i\in I\} , \{y'_j, j\in J\}) \text{\ is an arc in\ }  E(H_2) \text{\ and\ } x_i=x'_j \text{\ for all\ } i\in I, j\in J.
			\end{align*}

The Cartesian product is associative, commutative, and the trivial one-vertex
hypergraph $K_1$ without arcs serves as unit \cite{HOS-12, OstHellmStad11:CartProd}.
Thus, for arbitrary finitely many factors $\{H_i, i\in I\}$ the product $\Box_{i
\in I} H_i$ is well-defined, and each vertex $x\in V(\Box_{i \in I} H_i)$ is
properly ``coordinatized'' by the vector $(x_i)_{i\in I}$ whose entries are the
vertices $x_i$ of the factors $H_i$.

A nontrivial hypergraph $H$ is \emph{prime} with respect to the Cartesian
product if it cannot be represented as the Cartesian product of two nontrivial
hypergraphs. A \emph{prime factor decomposition (PFD)} of $H$ is a
representation as a Cartesian product $H=\Box_{i\in I} H_i$ such that all factors $H_i$,
$i\in I$, are prime and $H_i\not\simeq K_1$. Note, the number $k$ of  prime
factors of $H = ( V , E )$ is bounded by $log(| V |)$, since every Cartesian
product of $k$ non-trivial hypergraphs has at least $2^k$ vertices.

Two important results concerning the Cartesian products of hypergraphs are given now. 

\begin{lemma}[\cite{OstHellmStad11:CartProd}]
 The Cartesian product $H=\Box_{i=1}^n H_i$  of directed hypergraphs
	is connected if and only if all of its factors $H_i$ are connected.
\label{lem:Cart_connected}
\end{lemma}

\begin{thm}[\cite{OstHellmStad11:CartProd}]
	Connected (directed) hypergraphs have a unique prime factor decomposition
	with respect to the Cartesian product. 
	\label{thm:upfd}
\end{thm}

We will show, that the PFD of a hypergraph $H$ can be obtained
from the PFD of its 2-section $[H]_2$. For this the following lemma is crucial.

\begin{lemma}
If $\Box_{i\in I} H_i$ is an arbitrary factorization of $H$ it holds
that $[H]_2 = \Box_{i\in I} [H_i]_2$. 
\label{lem:Cart_2section}
\end{lemma}
\begin{proof}
	Since the Cartesian product is commutative and associative it suffices to 
	prove the statement for two factors. 
   Assume that $H=H_1\Box H_2$ and every vertex $x$ has coordinates $(x_1,x_2)$.
	Thus, there is an isomorphism
	$\phi:V(H)\to V(H_1\Box H_2)$ via $x\mapsto (x_1,x_2)$. 
	We show that $\phi$ is also an isomorphism 
	for the graphs $[H]_2$  and $[H_1]_2\Box [H_2]_2$.

	The edge $xy$ is contained in $E([H]_2)$ if and only if
	there is an arc $e\in E(H)$ with $x,y\in V(e) = t(e)\cup h(e)$ if and only if 
 	(i)  $x_1=y_1$ and $x_2y_2 \subset V(f) \in E(H_2)$ or 
	(ii) $x_2=y_2$ and $x_1y_1 \subset V(f) \in E(H_1)$ if and only if
 	(i)  $x_1=y_1$ and $x_2y_2 \in E([H_2]_2)$ or 
	(ii) $x_2=y_2$ and $x_1y_1 \in E([H_1]_2)$ if and only if
	the edge $\phi(x)\phi(y)$ is contained in $E([H_1]_2\Box [H_2]_2)$.
\end{proof}

Now, given the PFD of $H=\Box_{i\in I} H_i$, we can infer that
$[H]_2=\Box_{i\in I} [H_i]_2$. 
However, the factors $[H_i]_2$ might not be prime w.r.t. the Cartesian
\emph{graph} product and hence, $[H]_2$ might have more prime factors.
Since the PFD of $[H]_2 = \Box_{j\in J} G_j$ is unique it follows that the 2-section
of the prime factors $H_i$ of $H$ is a combination of the prime
factors $G_j$ of $[H]_2$, that is, 
$[H_i]_2 \simeq \Box_{j\in J'} G_j$, $J'\subseteq J$ for all 
$i\in I$.

Our algorithm will start with the PFD of $[H]_2=\Box_{j\in J} G_j$ 
w.r.t. the Cartesian product of undirected graphs and then tries to combine the respective 
prime factors of $[H]_2$ to reconstruct the prime factors of $H$. In other words, we need
to find suitable subsets $J'\subseteq J$ so that $[H_i]_2 \simeq \Box_{j\in J'} G_j$
and $H_i$ is a prime factor of $H$. 
To this end, we will introduce (pre-)coordinates and (pre)-layers.

\begin{defi}[(Factorization) Coordinatization]
Let $H$ be isomorphic to some product $\Box_{i \in I} H_i$, where
each factor $H_i$ has vertex set $\{1,\dots,l_i\}$. 
A \emph{factorization coordinatization} or \emph{coordinatization} for
short, is an isomorphism $\coord$ from $H$ to $\Box_{i \in I} H_i$. 
Thus, $\coord$ assigns to a vertex $v \in V(H)$ a vector of coordinates
$(v_i)_{i \in I}$ where $1 \leq v_i \leq l_i$ is a vertex in $V(H_i)$
\label{defi:coord}
\end{defi}

Hence, a coordinatization gives in an explicit way the information of the
underlying product structure of $H$. Hence, to find a factorization of $H$ one
can equivalently ask for a coordinatization of $H$, a fact that we will utilize
in our algorithm.  
Note that the coordinatization w.r.t.\ a given
product decomposition is unique up to relabeling the vertices in each factor
$H_i$. 

We will also need a notion which is similar to a coordinatization but 
is implied by a factorization of the $2$-section $[H]_2$ rather than a decomposition of $H$.

\begin{defi}[Pre-Coordinatization]
Let $H$ be a given hypergraph and assume that $[H]_2$ has a coordinatization 
$\coord\colon V([H]_2) \to \bigtimes_{i \in I} \{1,\dots,l_i\}$. 
Since $V([H]_2)=V(H)$ we infer that $\coord$ is a bijective 
map on $V(H)$ that assigns to each vertex $v\in V(H)$ a unique coordinate-vector 
$(v_i)_{i \in I}$ where $1 \leq v_i \leq l_i$. 
This map is called \emph{pre-coordinatization} of $H$. 
\label{defi:pre-coord}
\end{defi}

For convenience, we will usually omit the function $\coord$
and identify every vertex $v\in V(H)$ with its (pre-)coordinate vector, i.e., 
we will write $v = (v_i)_{i \in I}$ rather than $\coord(v)=(v_i)_{i \in I}$.

\begin{defi}[Layers and Pre-Layers]
Let  $H \simeq \Box_{i \in I} H_i$ with given respective coordinatization 
$v = (v_i)_{i \in I} \in V(H)$ and $I' \subseteq I$.
The \emph{$I'$-layer through $v$} (denoted by $H_{I'}^v$) with
respect to this coordinatization is the sub-hypergraph induced by the vertices
$\{u=(u_i)_{i \in I} \mid i \notin I' \implies u_i = v_i \}$, i.e., we fix all
coordinates except those contained in the set $I'$. 
Note, $H_{I'}^v \simeq \Box_{i\in I'} H_i$.

Analogously, as layers are defined by means of a coordinatization define
the \emph{pre-layers} by means of a pre-coordinatization.

For simplicity we write
$H_i^v$ instead of $H_{\{i\}}^v$ and $i$-(pre-)layer rather than $\{i\}$-(pre-)layers.
\label{defi:layer}
\end{defi}

For later reference, we need the following observation and lemma. 
If $H \simeq \Box_{i \in I} H_i$ and $e$ is an arc of $H$, 
then all vertices in $V(e)$ are contained in
the same $i$-layer $H_i^v$ for some $i \in I$ and $v\in V(H)$, 
i.e.\ they only differ in the $i$-th coordinate. The same is true
for pre-layers. 

\begin{lemma}
Let $H$ be a hypergraph and let $\coord$ be a pre-coordinatization of $H$. 
Then every arc $e$ of $H$ contains vertices of exactly one pre-layer w.r.t\ $\coord$, that is, 
all vertices in $V(e)$  only differ in the same $i$-th coordinate.
\label{lem:arc-in-prelayer}
\end{lemma}
\begin{proof}
Every hyperarc $e$ forms a complete subgraph $K_{|V(e)|}$
in $[H]_2$. Moreover, complete subgraphs must be contained entirely in one of the 
$i$-layers of $[H]_2$, as complete graphs are so-called S-prime graphs, see e.g.
\cite{Bresar:03,Hel-12,HGS-09,Klavzar:02}. Hence, each hyperarc is contained in
one $i$-pre-layer of $H$.
\end{proof}

Note, any isomorphism from $[H]_2$ to $\Box_{i\in I} [H_i]_2$ and thus,
a  pre-coordinatization of $H$, is a coordinatization of $H$ if and
only if $H$ has a factorization $\Box_{i\in I} H_i$. 
Lemma \ref{lem:Cart_2section} immediately implies that
every coordinatization of $H$ is also
a pre-coordinatization of $H$, while the converse is not true in general. 
On the other hand, we have the following result for so-called increments of coordinates.

\begin{defi}[Increments of Coordinates]
Given a pre-coordinatization $\coord\colon V(H) \to \bigtimes_{i \in I} \{1,\dots,l_i\}$,
of $H$ and a vertex $v=(v_1,\dots,v_i,\dots,v_k) \in V(H)$
we define $\inc(v,i)$ (w.r.t.\ $\coord$) as the vertex with 
coordinates $(v_1,\dots,v_i+1,\dots,v_k)$ where we set 
$v_i+1:=1$ if $v_i+1>l_i$. 

For an (ordered) set of vertices $W\subset V$ we define
the (ordered) set $\inc(W,i))=\{\inc(w,i)\mid w\in W\}$.

Finally, we denote for an arc $e=(t(e),h(e))$ its increment  
$(\inc(t(e),i), \inc(h(e),i))$ by $\inc(e,i)$. 
\end{defi}

\begin{lemma}
	Let $H =(V,E)$ be a hypergraph and $\coord\colon V\to \bigtimes_{i \in I} \{1,\dots,l_i\}$
	be a pre-coordinatization of $H$. 
	If for each arc $e\in E$ (where the vertices of $e$ differ only in the $j$-th coordinate),
	there is an arc $\inc(e,i) \in E$ for all $i\neq j$, 
	then $\coord$ is a coordinatization of $H$.
	\label{lem:prec-inc}
\end{lemma}
\begin{proof}
	By Lemma	\ref{lem:arc-in-prelayer}, 
	all vertices within one arc $e\in E$ differ in
	precisely one coordinate. 
	Let $e\in E$ be an arbitrary hyperarc and assume the vertices differ in the $j$-th coordinate. 

	Let $\mathcal{H}_j$ be the set of $j$-pre-layers contained in $H$. 
	Let $i\neq j$ be an arbitrary index $i\in I$. Assume that for each hyperedge $e$
	contained in some $j$-pre-layer $H(1)$ all ``incremental copies''
	$\inc(e,i)$ are also contained in $E$, then there is a homomorphism from
	$H(1) = \langle V(H(1))\rangle$ to $H(2)=\langle(inc(V(H(1)),i))\rangle$,
	where $H(2)$ corresponds to some other $j$-pre-layer. Assume that for all such
	``consecutive'' $j$-pre-layer there is a homomorphism from $H(l)$ to
	$H(l+1)$, $1\leq l\leq l_i-1$. By construction, after $l_i-1$ incremental
	steps we arrive at the $l_i$-th $j$-pre-layer $H(l_i)$ and hence,
	$H(1)=\langle(inc(V(H(l_i)),i))\rangle$. If there is an homomorphism from $H(l_i)$
	to $H(1)$, then there is trivially an isomorphism between all
	such $j$-pre-layers $H(1),H(2),\dots,H(l_i)\in\mathcal H_j$. 
	Thus, if for all arcs $e\in E$, where the
	vertices of $e$ differ precisely in this $j$-th coordinate, there is a hyperarc
	$\inc(e,i) \in E$ for all $i\neq j$, then there isomorphism between all
	$j$-pre-layers contained in $\mathcal H_j$ for this fixed $j\in I$. 

	If this is true for \emph{all} arcs $e\in E$, and thus, for 
	\emph{all} $i$-pre-layers with $i\in I$, then \emph{all} 
	such $i$-pre-layers are isomorphic for each $i\in I$. 

	In particular, we can define for vertices $v,w$ and an index $i\in I$
	the map $g_i^{vw}\colon H_j^v \to H_j^w$ which maps every vertex in $H_i^v$ to the
	unique vertex in $H_i^w$ with the same $i$-coordinate. 
	By the preceding arguments, 
	for each $i\in J$ the map $g_i^{vw}$ is an isomorphism 
	between the $i$-pre-layers in $H$ for all $v, w \in V(H)$. 

	Finally, assume for contradiction that
	$\coord$ is a not a coordinatization	
	and hence,  $\coord$ is not an isomorphism from $H$ to any product 
	$\Box_{i \in I} H_i$. Hence, there must be some $i\in I$ such
	that not all $i$-layers are isomorphic by means of 
	$g_i^{vw}$, a contradiction.
\end{proof}

Lemma \ref{lem:arc-in-prelayer} allows defining an equivalence relation
$R_\coord$ on the hyperedge set $E(H)$ for a given pre-coordinatization
$\coord\colon V(H) \to \bigtimes_{i \in I} \{1,\dots,l_i\}$ of $H$, as follows:
$(e,f)\in R_\coord$ if $e\in H_i^v$ and $f\in H_i^w$ for some $i\in I$ and
$v,w\in V(H)$. In other words, $e$ and $f$ are in relation $R_\coord$ if they
are both contained in the $i$-pre-layers for the same fixed $i\in I$. Note, in
case that $\coord$ is a coordinatization the relation $R_\coord$ is also known
as product relation, that is, each equivalence of $R_{\coord}$ contains the
hyperedges of all copies of some (not necessarily prime) factor of $H$. In order
to avoid confusion, we sometimes write that $R_{\coord}(H)$ to indicate that
$R_{\coord}$ is defined on the edge set of $H$.

Given two
pre-coordinatizations $\coord_1$ and $\coord_2$, we say that $\coord_1$ is finer
than $\coord_2$, while $\coord_2$ coarser than $\coord_1$ if $R_{\coord_1}$ is
finer than $R_{\coord_2}$.
We can immediately infer the next result. 

\begin{lemma}
Let $H$ be a hypergraph and let $\coord_1$ and $\coord_2$ be
pre-coordinatizations of $H$. Then $\coord_1$ is finer than $\coord_2$ if and
only if the factorization of $[H]_2$ corresponding to $\coord_2$ can be obtained
from the factorization $\Box_{i \in I} G_i$ corresponding to $\coord_1$ by
combining some of these factors, i.e., $I$ can be partitioned into 
$I_1,\dots,I_l$ so that the pre-coordinatization $\coord_2$ is 
an isomporphism between $[H]_2$ and $\Box_{j=1}^l G'_j$ where 
$G'_j \simeq  \Box_{i \in I_j} G_i$.
\label{lem:prec-partition}
\end{lemma}
\begin{proof}
 By definition $\coord_1$ is finer than  $\coord_2$ if and only if
 $R_{\coord_1}$ is finer than $R_{\coord_2}$. 
 By definition of the Cartesian product and by Theorem \ref{thm:upfd}
 this is the case if and only if every factor in the factorization 
 corresponding to $\coord_2$ 
 is a combination of the factors in the coordinatization
 $\coord_1$.
\end{proof}

\begin{rem}
By Lemma \ref{lem:Cart_2section}, every coordinatization of $H$ is also a
pre-coordinatization of $H$. This implies together with Lemma
\ref{lem:prec-partition} that \emph{every} pre-coordinatization of $H$ can be
achieved by starting with the pre-coordinatization w.r.t.\ the prime
factorization of $[H]_2$ and then combining the corresponding pre-layers of $H$
to obtain the layers w.r.t.\ the prime factorization of $H$. In other words, one
needs to find a partition $\biguplus_{j \in J} I_j$ of the index set $I$ and
then combine all pre-layers corresponding to indices in the same part into one.
\label{rem:combine}
\end{rem}

As we shall see later, in our algorithm we will \emph{only} check increments w.r.t.\ the
pre-coordinatization $\coord'$ coming from the PFD of $[H]_2$. 
However, we have to prove that this is indeed sufficient (Theorem \ref{thm:correctness-pfd}).
In order to apply Lemma
\ref{lem:prec-inc} to validate whether we end up with a coordinatization $\coord$ of 
$H$ we would need to check increments with respect to this coarser
(pre-)coordinatization $\coord$. Now one might hope that increments with respect to the
coarser pre-coordinatization are automatically increments with respect to the
finer pre-coordinatization $\coord'$ or that at least the coarser pre-coordinates can be
chosen in a suitable way. However, this is not the case as the following example
shows.

Assume that at some point we need to combine $i$- and $j$-pre-layers of sizes
$3$ and $4$ respectively. The resulting $k$-pre-layers with respect to the new
coordinatization will each contain $12$ vertices labeled $1, \ldots , 12$. We now
claim that no matter how we assign the new labels, there is always at least one
increment $\inc(\cdot,k)$ which is not an increment $\inc(\cdot,i)$ or
$\inc(\cdot,j)$. Assume for a contradiction that all increments $\inc(\cdot,k)$
were either of the form $\inc(\cdot,i)$ or $\inc(\cdot,j)$. By applying
$\inc(\cdot,k)$ recursively $12$ times to a vertex, we end up at the same vertex
again. This means, that we have applied $\inc(\cdot,i)$ a number of times which
must be divisible by $3$ and $\inc(\cdot,j)$ a number of times which must be
divisible by $4$. However, no suitable multiples of $3$ and $4$ add up to $12$.

The latter example shows that the single check of increments with respect to the
PFD of $[H]_2$ is not sufficient to invoke Lemma \ref{lem:prec-inc} to conclude
that some coarser pre-coordinatization is indeed a coordinatization. For this
purpose,  we need the following additional lemma.

\begin{lemma}
	Let $H =(V,E)$ be a hypergraph, let $\coord_1,\coord_2$ be
	pre-coordinatizations of $H$ such that $\coord_1$ is finer than $\coord_2$.
	Let $\inc_1$ and $\inc_2$ the respective increment maps. Assume that for each
	arc $e\in E$ (where the vertices of $e$ differ only in the $j$-th coordinate
	w.r.t.\ $\coord_2$), there is an arc $\inc_1(e,i) \in E$ for all $i\notin
	I_j$ (where $I_j$ is defined as in Lemma \ref{lem:prec-partition}). 
	Then there is an arc $\inc_2(e,k)$
	for all $k \neq j$ and hence $\coord_2$ is a coordinatization of $H$.
	\label{lem:prec-coarsen-inc}
\end{lemma}

\begin{proof}
The coordinates of the vertices in $\inc_2(e,k)$ w.r.t.\ $\coord_1$ can be
obtained from those of vertices in $e$ by only changing coordinates outside
$I_j$. This can be achieved by successive applications of $\inc_1(\cdot ,i)$ for
$i\notin I_j$. Since we started at an edge (namely $e$) and each of those
applications takes edges to edges we also end at an edge $\inc_2(e,k)$. Hence, 
Lemma \ref{lem:prec-inc} implies that $\coord_2$ is a coordinatization.
\end{proof}

\section{PFD-algorithm for Directed Hypergraphs}

\subsection{Workflow}

We give here a summary of the workflow of the algorithm to compute the
prime factor of connected directed hypergraphs.
The top-level control structure is summarized in Algorithm \ref{alg:pfd}
\PFDalg\ in which the subroutines 
\texttt{Preprocessing} (Alg.\ \ref{alg:pre}) and
\texttt{Combine} (Alg.\ \ref{alg:combine}) are used. 

\begin{algorithm}[htbp]
\caption{\PFDalg} 
\label{alg:pfd}
\begin{algorithmic}[1]
\small
    \STATE \textbf{INPUT:} A hypergraph $H=(V,E)$;
	 \STATE $H\gets\ $\texttt{Preprocessing}($H$); \COMMENT{Now, $H=(V_{\lex},E_{\lex})$ and vertex coordinates are known.}	 \label{alg:prepr}
	 \STATE Let $G_{\aux} = (\{1,\dots,k\}, \emptyset)$ \COMMENT{Here, $k$ is the number of prime factors of $[H]_2$.};\label{alg:init-Gaux}
	\FOR {each $e\in	 E_\lex$} \label{alg:for-start}
			\STATE Let $j\in \{1,\dots,k\}$ be the \emph{(unique)} coordinate where distinct $x,y\in V(e)$ differ; \label{alg:com1}
			\FOR{$i\in I=\{1,\dots,k\}\setminus\{j\}$} \label{alg:for2-begin}
					\IF {$\inc(e,i)\notin E_\lex$} \label{alg:if}
						\STATE add edge $ij$ to  $G_{\aux}$;
					\ENDIF
			\ENDFOR\label{alg:for2-end}
	 \ENDFOR \label{alg:for-end}
		\STATE \texttt{Combine}($H=(V_{\lex},E_{\lex}), G_\aux$); \label{alg:call-combine}
	 \STATE \textbf{OUTPUT:} PFD $\Box_{i=1}^n H_i$ of $H$;
\end{algorithmic}
\end{algorithm}
%
\begin{algorithm}[htbp]
\caption{\texttt{Preprocessing}} 
\label{alg:pre}
\begin{algorithmic}[1]
\small
    \STATE \textbf{INPUT:} A connected hypergraph $H=(V,E)$;
	 \STATE Compute $\Box$-PFD of $[H]_2 = \Box_{i=1}^k G_i$ and vertex-coordinates with the Imrich-Peterin-Algorithm \cite{Imrich07:linear}; \label{alg:2sect-pfd}
	 \STATE Compute the  list  $V_{\lex}$ of lexicographic ordered vertices (w.r.t.\ their coordinates);  \label{alg:lexV}
	 \FOR {each $e\in E$} \label{alg:for1-begin}
		\STATE Reorder $t(e)$ and $h(e)$ w.r.t.\ the lexicographic order of the vertices;
	 \ENDFOR \label{alg:for1-end}
	 \STATE Compute the list $E_{\lex}$ of lexicographic ordered arcs; 
			  w.r.t.\ to the lexicographic ordered sets $t(e)$ and then $h(e)$;
	\label{alg:lexE}
    \STATE \textbf{return} $(V_{\lex},E_{\lex})$ with respective vertex-coordinates;
\end{algorithmic}
\end{algorithm}
%
%
\begin{algorithm}[htbp]
\caption{\texttt{Combine}} 
\label{alg:combine}
\begin{algorithmic}[1]
\small
    \STATE \textbf{INPUT:} A hypergraph $H=(V,E)$ with pre-coordinates w.r.t\ the PFD of $[H]_2 $, a graph $G_\aux$;
		\STATE Compute connected components $I_1,\dots,I_l$ of $G_{\aux}$; \label{alg3:conn}
		\FOR {each $e\in	 E_\lex$} \label{alg3:fora}
			\STATE Let $j\in \{1,\dots,k\}$ be the \emph{(unique)} coordinate where distinct $x,y\in V(e)$ differ;
			\STATE Let $I_l$ be the connected component containing vertex $j$;
			\STATE Assign color $l$ to hyperedge $e$;
		\ENDFOR \label{alg3:fore}
		\STATE compute $[H]_2$ where edges $ij$ obtains the unique color of $e$ where $i,j\in V(e)$; \label{alg3:h2}
		\STATE compute coordinates of all vertices in $V$ in $[H]_2$; \COMMENT{This pre-coordinatization is the PFD-coordinatization of $H$} \label{alg3:coord}
	 \STATE \textbf{return} $\Box$-PFD $\Box_{i=1}^n H_i$ of $H$;
\end{algorithmic}
\end{algorithm}

As input of \PFDalg\  a connected hypergraph $H=(V,E)$ is
expected. First of all, subroutine \texttt{Preprocessing} is called. Here, the
PFD of $[H]_2 = \Box_{i=1}^k G_i$ and the respective coordinatization $\coord$
of $[H]_2$ is computed by application of the algorithm of Imrich and Peterin
\cite{Imrich07:linear}.
 Then the
vertices, the vertices within the arcs and the arcs are ordered in
lexicographic order. This helps to achieve the desired time-complexity in later steps.

By definition, coordinatization $\coord$ of $[H]_2$ is a pre-coordinatization of
$H$. By construction of $\coord$ and Lemma \ref{lem:Cart_2section}, the
pre-coordinatization $\coord$ is at least as fine as the coordinatization of $H$
w.r.t.\ its PFD. By Remark \ref{rem:combine} it suffices to find a suitable
partition of $I=\{1,\dots,k\}$ to derive the prime factors of $H$. 
To this end, we initialize in Line
\ref{alg:init-Gaux} of Algorithm\ \ref{alg:pfd} the auxiliary graph $G_{\aux}$,
where each vertex $i$ represents an element of $I$. The edge set is left empty.
We might later add edges in order trace back which equivalence classes of
$R_\coord$ have to be combined, i.e., all vertices within one connected components of $G_{\aux}$
will then be in one class of the respective partition of $I$.

We continue to check in the for-loop in Line \ref{alg:for-start}-\ref{alg:for-end}
of Algorithm\ \ref{alg:pfd}, if for each arc $e\in E_{\lex}$ that is contained in
some $j$-layer its ``copies'' are also contained in ``incremental-neighboring''
$j$-layers, i.e., we check if $\inc(e,i)\in E_\lex$. If this is not the case, 
then we add the edge $ij$ to $G_{\aux}$. 
Finally, we use the information of the connected components $I_1,\dots,I_l$ of 
$G_{\aux}$ that partition the set $I$ in order to determine the prime factors
of the given hypergraph $H$. To this end, the subroutine \texttt{Combine} is called
and an edge-colored 2-section $[H]_2$ is computed. That is, each edge $e\in E([H]_2)$ that
is contained in the copy of factor $G_i$ with $i\in I_s$ obtains color $s$. 
In other words, all prime factors $G_j, j\in I_s$ of $[H]_2$ are combined
to a single factor of $[H]_2$ and the edges in the respective $I_s$-layers obtain color $s$. 
W.r.t.\ this coloring it is possible to efficiently determine 
new vertex coordinates in  $[H]_2$ which is then a factorization coordinatization
of $H$. We will show, that this leads to a ``finest'' coordinatization of $H$
and hence, to the prime factors $H_1,\dots,H_r$ of $H$.
	
For illustrative examples see Figures \ref{fig:exmpl1} and \ref{fig:exmpl2}.

\begin{table}[htbp]
\begin{figure}[H]
	\centering{
	\includegraphics[bb=0 0 491 144, scale=.9]{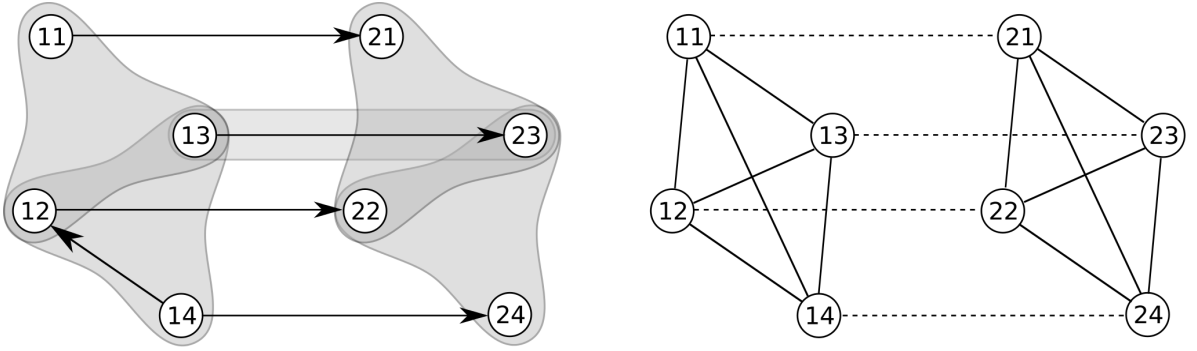}
	}
	\caption{The prime hypergraph $H$ (left-hand side) with arc set 
				$E(H)=\{ {e_1=(\{11,12,13\},\{12,13,14\})},KA:
							{e_2=(\{21,22,23\},\{22,23,24\})},
							e_3=(\{14\},\{12\}),
							e_4=(\{11\},\{21\}),
            			{e_5=(\{12\},\{22\})},
							{e_6=(\{13,23\},\{13,23\})},
							{e_7=(\{14\},\{24\})}, 
							{e_8=(\{13\},\{23\})}
						 \}$ 
				admits a non-trivial prime factorization of its 2-section $[H]_2$ (right-hand side)
				into $K_4\Box K_2$. The vertices of $H$ are labeled w.r.t.\ its pre-coordinatization 
				given by the coordinatization of $[H]_2$. For the first arc $e_4$ of
				the lexicographic ordered arc set $E_\lex  = \{e_4,e_1,e_5,e_8,e_6,e_3,e_7,e_2\}$
				we can observe that the increment $\inc(e_4,2) = ((\{12\}, \{22\}) = e_5 \in E_\lex$.
				Analogously, $\inc(e_1,1)=e_2$, $\inc(e_5,2)=e_8$, and $\inc(e_8,2)=e_7$ 
				are all contained in $E_\lex$. However, when we arrive at the hyperedge $e_6$
				we obtain that $\inc(e_6,2) = (\{14,24\},\{14,24\})$ is not an arc of $E_\lex$. 
				Hence, applying Algorithm \ref{alg:pfd} would lead to an edge $12$ in $G_\aux$, resulting in a connected auxiliary graph
				and, $H$ would be determined as prime. 
				A second example for an arc $e\in E_\lex$ with $\inc(e,i)\notin E_\lex$ is the edge $e=e_7$.
				}
	\label{fig:exmpl1}
\end{figure}

\begin{figure}[H]
	\centering{
	\includegraphics[bb=0 0 409 300, scale=.9]{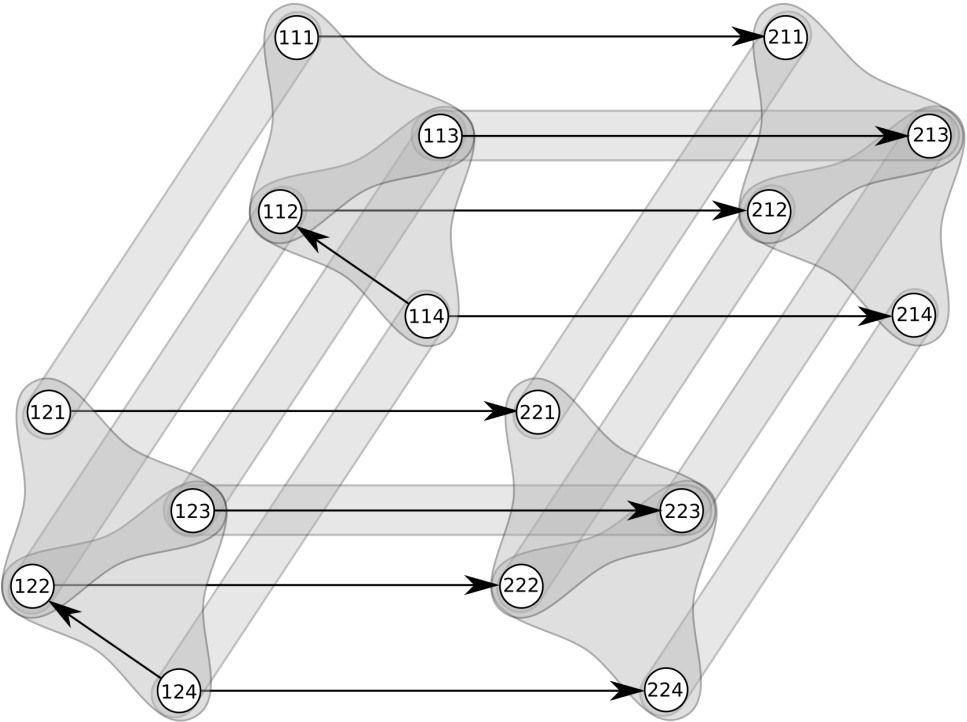}
	}
	\caption{The non-prime hypergraph $H = H_1 \Box H_2$ 
				is the product of the directed prime hypergraph $H_1$ in Fig.\ \ref{fig:exmpl1} and
				an undirected hypergraph $H_2$ with two vertices and one hyperedge. 
				The vertices of $H$ are labeled w.r.t.\ its pre-coordinatization 
				given by the coordinatization of $[H]_2$. The auxiliary graph $G_\aux$ is initialized
				as the graph with three vertices and empty edge-set in Algorithm \ref{alg:pfd}.
				While the increment $\inc(e,2) = \{(123,223),(123,223)\}$ of the arc $e=\{(113,213),(113,213)\}$ is 
				still contained in $E_\lex$, the increment $\inc(e,3) = \{(124,224),(124,224)\}$ is not. 
				Hence, the edge $13$ is added to $G_\aux$. Since the increments of all hyperedges of the form
				$\{(i1j,i2j),(i1j,i2j)\}$ are contained in $E_\lex$, no further edges will be added 
				to $G_\aux$. Hence, the sub-hypergraph induced by vertices with identical 2nd coordinate, 
				i.e., the $\{1,3\}$-layers, constitute the copies the prime factor $H_1$, while the
				sub-hypergraphs induced by vertices with identical 1st and 3rd coordinate, i.e, the $2$-layers,
				are copies of the prime factor $H_2$.
				}
	\label{fig:exmpl2}
\end{figure}
\end{table}

\subsection{Correctness}

We are now in the position to prove the correctness of the algorithm \PFDalg,
summarized in the following theorem.

%

\begin{thm}
	Algorithm \ref{alg:pfd} is sound and complete.
	\label{thm:correctness-pfd}
\end{thm}
\begin{proof}
	Given a hypergraph $H=(V,E)$. We start with a preprocessing and call in Algorithm\
	\ref{alg:pfd} the Algorithm\ \ref{alg:pre}. Here, the PFD of $[H]_2 = \Box_{i=1}^k
	G_i$ and the respective coordinatization $\coord$ of $[H]_2$ is computed.
	This coordinatization $\coord$ is by definition a
	pre-coordinatization of $H$. Since $[H]_2$ is an undirected graph, it is
	allowed to apply the algorithm of Imrich and Peterin \cite{Imrich07:linear}.
	Finally, the vertices, the vertices within the arcs and the arcs are
	ordered in lexicographic order. The latter task is not important for the
	correctness of the algorithm, but for the time-complexity that we will
	consider later on. 

	We are now in Line \ref{alg:init-Gaux} of Algorithm\ \ref{alg:pfd}. 
	By construction of $\coord$ and Lemma \ref{lem:Cart_2section}, the pre-coordinatization $\coord$
	is at least as fine as the coordinatization of $H$ w.r.t.\ its PFD. By Lemma
	\ref{lem:prec-partition} it suffices to find a suitable partition of
	$I=\{1,\dots,k\}$. To this end, we initialize in the auxiliary graph
	$G_{\aux}$ where each vertex $i$ represents an element of $I$. The edge set
	is left empty. We might later add edges in order trace back which equivalence
	classes have to be combined, i.e., all vertices within one connected
	components of $G_{\aux}$ will then be in one class of the respective
	partition of $I$.
	
	Now consider the for-loop in Line \ref{alg:for-start}-\ref{alg:for-end}. For
	each $e\in E_{\lex}$ we check in which coordinates the vertices in $V(e)$
	differ. Since $\coord$ is a pre-coordinatization and by Lemma \ref{lem:arc-in-prelayer}, 
	this is exactly one	coordinate for each hyperarc. Let $e\in E_{\lex}$ be a chosen arc and assume
	that all vertices in $V(e)$ differ in the $j$-th coordinate. Now, it is
	checked if for arc $e$ its ``copies'' are contained in each
	$G_j^w$-layer where $w\in inc(V(e),i)$. If for some arc $e\in E_{\lex}$ we
	observe that there is no hyperedge $\inc(e,i) = (\inc(t(e),i), \inc(h(e),i)\in E_\lex$ then
	there is no ``copy'' of $e$ in some $j$-pre-layer through $w$ with $w\in
	inc(V(e),i)$. In this case we add the edge $ij$ to $G_\aux$ if not already
	set. The latter tasks are repeated for all hyperarcs $e\in E_{\lex}$.

	Finally, in Line \ref{alg:call-combine} the Algorithm\ \ref{alg:combine}
	is called. The task of this subroutine is to combine the pre-coordinates
	and thus, the pre-layers in order to determine the layers of the 
	final prime-factors of $H$.
   Let $I_1,\dots,I_r$ be the connected components of $G_\aux$.
	Clearly, $\Pi=\{I_1,\dots, I_r\}$ is a partition
	of $I$. Let each $I_j$ having $l_j$ elements. 
	Lemma \ref{lem:prec-partition} and Remark \ref{rem:combine} imply that
   $\widehat{\coord}\colon V(H) \to \bigtimes_{l=1}^r \{1,\dots,l_r\}$ 
	is  a pre-coordinatization of	$H$. 
	It remains to show that 
	\begin{enumerate}
		\item[(1.)]  $\widehat{\coord}$ is a coordinatization and 
		\item[(2.)]  $\widehat{\coord}$ is at least as fine as the coordinatization given by the PFD of $H$. 
	\end{enumerate}

	\emph{Claim (1.):}
	By construction, all $e\in E_\lex$ where the vertices differ in the $i$-th coordinate
	w.r.t.\ $\coord$ are now contained in some $I_s$-layer where $i\in I_s\in \Pi$. 
	Moreover, for all $e$ in some $I_s$-layer the increments $\inc(e,j)$ with
	$j\neq i$ and $j\notin I_s$ must be contained in $H$, as otherwise we would
	have added the edge $ij$ to $G_\aux$ and hence, $j\in I_s$. 
	As the latter is true for all $I_s$-layers contained in $\Pi$
	we can apply Lemma \ref{lem:prec-coarsen-inc} and conclude that 
	 $\widehat{\coord}$ is a coordinatization of $H$.

	\emph{Claim (2.):}
	Given the pre-coordinatization $\coord$ of $H$. 
	By construction of $\coord$ and Lemma \ref{lem:Cart_2section},
   $\coord$
	is at least as fine as the coordinatization of $H$ w.r.t.\ its PFD.
	Thus, there is a partition $\Pi'=\{I'_1,\dots, I'_t\}$ of $I$
	 w.r.t.\ the PFD of $H$.
	It remains to show that if there are two indices $i,j \in I_s\in \Pi$,
   then $i,j$ are also contained in the same class of $\Pi'$.
	If $i,j \in I_s\in \Pi$ then they are in same connected component
	$C_s$ of $G_\aux$. Hence, it suffices to consider pairs 
	$i,j \in I_s\in \Pi$ that are connected by an edge. 
	Assume, for contradiction that $i,j$ are in different classes of $\Pi'$. 
	W.l.o.g. let $i\in I'_1$ and $j\in I'_2$. 
	Moreover, let $H=H_{I'_1}\Box H_{\cup_{l\geq2} I'_l}$  
	Hence, for \emph{all} $l\in I'_1$ and thus, in particular for $l=i$
	it holds that for all arcs $e$ in some $I'_2$-layer there is an arc 
	$(\inc(t(e),i), \inc(h(e),i)\in E_\lex$. 
	The same holds with the role of $i$ and $j$ switched. 
	However, in this case we would not add the edge $ij$ to $G_\aux$,
	a contradiction.

	To finish the PFD-computation we have to compute $\widehat\coord$.
	To this end, we compute the 2-section $[H]_2$ with edges $xy$ colored
	with color $j$ whenever $x$ and $y$ are contained in some edge $e$
	that is contained in some $j$-layer of $H$.  
	Lemma \ref{lem:Cart_2section} implies that  $\widehat\coord$
   is also 	a coordinatization of $[H]_2$ and hence, 
	all edges with same color $j$ in $[H]_2$ are contained in the same
	equivalence class of $R_{\widehat\coord}([H]_2)$.
	In other words,  $(e,f) \in R_{\widehat\coord}(H)$ if and only if
	 $(xy,uv) \in R_{\widehat\coord}([H]_2)$ for all distinct $x,y\in V(e)$
	and distinct $u,v\in V(f)$. By construction, $R_{\widehat\coord}([H]_2)$ is a
	a product relation of $[H]_2$, and thus we can apply  again 
	a method proposed the by Imrich and Peterin (cf.\ Theorem 5.1. in \cite{Imrich07:linear}), 
	in order to obtain the desired coordinates and hence, $\widehat\coord$.
\end{proof}

\subsection{Time Complexity}

In order to prove the time-complexity results, we first give the following lemma. 

\begin{lemma}
Let $H$ be a hypergraph, let $[H]_2 = \Box_{i=1}^k H_i$ be a factorization of its $2$-section into $k$ factors, and let $m$
and $n$ be the number of arcs and vertices of $H$, respectively. 
Then for any $l \in \mathbb N_0$ it holds that $k^l \log m = O(n)$.
\label{lem:bound}
\end{lemma}
\begin{proof}
If $m_i$ and $n_i$ are the numbers of arcs and vertices of the factors then
$m_i \leq (2^{n_i})^2$, since we have to consider tail and head independently. 
Let $N$ be the maximum number of vertices of a
factor. Then we have 
\[
	m 
	= \sum_{i=1}^k m_i \prod_{\substack{j=1 \\j \neq i}}^k n_j 
	\leq \sum_{i=1}^k 2^{n_i} \prod_{j=1}^k n_j
	\leq k \cdot 2^N \prod_{j=1}^k n_j.
\]
Taking logarithms on both sides of the inequality gives
\[
	\log m 
	\leq \log k + N + \sum_{j=1}^k \log(n_i)
	\leq k + N + k\cdot N
	\leq c \cdot k \cdot N
\]
for some suitable constant $c$.

On the other hand by bounding the size of every factor except the biggest one
from below by $2$ we get $n = \prod_{i=1}^k n_i \geq N \cdot 2^{k-1}$. 
Clearly, $N \cdot 2^{k-1}\geq c' \cdot N \cdot k^l$ for some suitable constant $c'$
depending on $l$. Together with the estimate for $\log m$ this proves the lemma.
\end{proof}

The next two lemmas are concerned with the time-complexity of the subroutines
\texttt{Preprocessing} and \texttt{Combine}.

\begin{lemma}
	Let $H = (V,E)$ be a connected hypergraph with $|V|=n$, $|E|=m$ and rank $r$.
	Then Algorithm \ref{alg:pre} performs in $O(r^2mn)$ time.
	If we assume that $H$ has
	bounded rank, then  Algorithm \ref{alg:pre} has time-complexity $O(m\log^2(n))$. 
	\label{lem:runtime-pre}
\end{lemma}
\begin{proof}
	In Line \ref{alg:2sect-pfd}, the first task is the computation of the 2-section $[H]_2$. 
	To this end, we initialize an
	adjacency list  $N[1],\dots,N[n]$ with empty entries,
	 which can be done in $O(n)$ time. We add
	for each arc $e\in E$ and each pair $\{i,j\}\in \binom{V(e)}{2}$ the vertex
	$i$ to $N[j]$ and $j$ to $N[i]$, if these vertices are not already contained
	in the respective adjacency lists. Hence, we must check whether $i\in N[j]$ or not.
	To this end, assume that $N[j]$ is already ordered. Hence we need
	$O(\log(n))$ comparisons to verify if $i\in N[j]$. If this is not the case, 
	the vertex $i$ is added to $N[j]$ on the respective position so that $N[j]$ stays sorted.
	Analogously, we add $j$ to $N[i]$, whenever $j$ is not contained in $N[i]$.
	As for each arc $e\in E$ there are at most $\binom{r}{2} = O(r^2)$ pairs 
	$\{i,j\}$ and	for each such pair we have $O(\log(n))$ comparisons we end in a time-complexity of
	$O(n+mr^2 \log(n))$ to create the adjacency list 
	$N[1],\dots,N[n]$. These lists serve
	than as input for the algorithm of Imrich and Peterin which computes the PFD
	of the 2-section in $O(|E([H]_2)|+n)$ time. Since $[H]_2$  is connected 
	and thus, $[H]_2$ has at least $|V|-1$ edges,  
	the PFD algorithm runs in fact in $O(|E([H]_2|)=O(mr^2)$ time. Hence, the total
	time complexity of Line \ref{alg:2sect-pfd} of Algorithm\ \ref{alg:pre} is $O(n+ mr^2
	\log(n) + mr^2) = O(mr^2 \log(n))$. 

	In what follows, let $k$ be the number of factors of $[H_2]$
	and let each $v\in V$ be identified with its respective (pre-)\ 
	coordinate vector $(v_1,\dots,v_k)$ computed by the  Imrich-Peterin-Algorithm.
	Note, $k$ is bounded by $\log(n)$.

	In Line \ref{alg:lexV}, the list $V$ of vertices is reordered in lexicographic
	order w.r.t.\ the vertex coordinates, i.e., $v<w$ if there is some
	$i\in\{1,\dots,k\}$ with $v_j\leq w_j$ for all $j\in\{1,\dots,i-1\}$ and
	$v_i<w_i$. This task can be done in $O(n\log(n)k) = O(n\log^2(n))$. Since
	$mr\geq n$ we obtain that $O(n\log^2(n)) = O(mr\log^2(n))$. This new ordered
	vertex list is called $V_\lex$.	
	
	We are now concerned with the for-loop in Line \ref{alg:for-start}. For each
	hyperedge $e\in E$ we reorder the vertices of its head and tail w.r.t.\ to the
	order of the vertices in $V_\lex$. Each hyperedge contains at most $r$ vertices and
	hence, this task can be done in $O(r\log(r))$ time. Therefore, the entire
	for-loop (Line \ref{alg:for1-begin} - \ref{alg:for1-end}) takes
	$O(mr\log(r))$ time. 

	Finally, the arcs are reordered w.r.t. the lexicographic ordered sets $t(e)$
	and $h(e)$. We say $e<f$ if $t(e)<t(f)$ or $t(e)=t(f)$ and $h(e)<h(f)$,
	whereby the tails, resp., heads are compared w.r.t. the lexicographic order
	of their vertices. To determine if $t(e)<t(f)$ or $h(e)<h(f)$ for some arcs
	$e,f\in E$, the at most $2r$ pairs of vertices must be compared, whereby the
	comparison of each such pair can be done in $O(1)$ time, since the vertices
	are already ordered in the tails and heads. The reordering of the arcs need
	than $O(m\log(m))$ comparisons, where each comparison can be done in $O(r)$
	time, by the preceding arguments. Hence, the creation of $E_\lex$ takes $O(r
	m\log(m))$ time. By Lemma \ref{lem:bound} this is $O(rmn)$. Moreover, if we
	assume that the rank $r$ is bounded, then $m \leq \sum_{i=1}^r \binom nr \leq rn^r = O(n^r)$. Hence $O(\log(m))
	= O(\log(n^r)) = O(r\log(n))=O(\log(n))$. 
	In this case the time complexity for
	determining $E_\lex$ is $O(m\log(m)) = O(m\log(n))$. 	

	Taken together the latter arguments, we end in overall time complexity for
	Algorithm \ref{alg:pre} of $O(r^2mn)$ 
	and if the rank $r$ is bounded with $O(m\log^2(n))$. 
\end{proof}

\begin{lemma}
	Let $H = (V,E)$ be a connected hypergraph with $|V|=n$, $|E|=m$ and rank $r$.
	Moreover, assume that the graph $G_\aux$ has $k$ vertices and $m'$ edge with 
	$k\leq \log(n)$. 
	Then Algorithm \ref{alg:combine} preforms in $O(mnr^2)$ time.
	If we assume that $H$ has
	bounded rank, then  Algorithm \ref{alg:combine} has time-complexity $O(m\log^2(n))$. 
	\label{lem:runtime-post}
\end{lemma}
\begin{proof}
   Determining the connected components of $G_\aux$ in Line \ref{alg3:conn}
	can be done in $O(k+m')=O(\log(n)+\log^2(n))$ time by application of 
   the classical breadth-first search. 
	While doing this, we will in addition record in $O(1)$ time for each vertex 
	in which connected component it is contained.
	Let $I_1,\dots,I_l$ be the connected components of $G_\aux$.
	
	For each of the $m$ arcs we have to find the indices where the vertices of
	the particular arc differs. To this end, it suffices to take any two vertices
	$x$ and $y$ of $V(e)$ and to compare their $k$ coordinates which takes $O(k)$
	time. Let $j$ be the coordinate where the two vertices differ. We need to
	check in which of the connected components $I_s$ the vertex is contained in,
	which can be done in $O(1)$ time, since we have already recorded for each
	vertex of $G_\aux$, in which component it is contained in. 
	Now, the color for each arc can be 
	recorded in $O(1)$ time.
   Hence, the for-loop (Line \ref{alg3:fora}-\ref{alg3:fore}) has
	overall-time complexity $O(mk) = O(m\log(n))$.

	To compute the 2-section in Line \ref{alg3:h2} with colored edges we
	initialize an extended adjacency list $N[1],\dots,N[n]$ where whenever we add
	some $i\in N[j]$ we also record the respective unique color of $ij$ as a 2nd
	parameter. Recording this parameter can be done in $O(1)$ time, as for each
	arc $e\in E$ it is known which color it has. Hence, we can argue analogously
	as in the proof of Lemma \ref{lem:runtime-pre}, and state that the 2-section
	with additionally colored edges
	can be computed in $O(n+mr^2 \log(n))$ time. 

	Finally, the vertex-coordinates in $[H]_2$ can be computed in $O(m')=O(mr^2)$
	time, see Theorem 5.1. in \cite{Imrich07:linear}. 
	
	Hence the overall-time complexity of Algorithm \ref{alg:combine} 
	is $O(\log^2(n) + m\log(n) + n+mr^2 \log(n) + mr^2) = O(n+m\log^2(n)r^2)$. 
	Since $mr\geq n$ and $\log^2(n) = O(n)$,
	the latter can be expressed as $O(mnr^2)$. 
   If we assume in addition that the rank $r$ is bounded we get 
	$O(n+m\log^2(n)r^2) =O(mr+m\log^2(n)r^2) = O(m\log^2(n))$.
\end{proof}

We are now in the position to determine the time-complexity of algorithm 
\PFDalg.

\begin{thm}
	Let $H = (V,E)$ be a connected hypergraph with $|V|=n$, $|E|=m$ and rank $r$.
	Then Algorithm \ref{alg:pfd} computes the PFD of $H$ in $O(mnr^2)$ time. 
	If the rank $r$ is bounded the time-complexity of Algorithm \ref{alg:pfd} 	
	is $O(m\log^2(n))$.
	\label{thm:runtime-pfd}	
\end{thm}

\begin{proof}
	We suppose both the vertices and the hyperarcs of $H$ implemented
	as integers and $E$ implemented as an $m\times 2$ array, 
	where each entry $E[e,i]$ contains the list of vertices in $t(e)$
	if $i=1$ and $h(e)$	if $i=2$.
	In Line \ref{alg:prepr} we call \texttt{Preprocessing}($H$) which takes  
   $O(r^2mn)$ time and if $r$ is bounded $O(m\log^2(n))$ time
	(Lemma \ref{lem:runtime-pre}). 

	In what follows, let $k\leq\log(n)$ be the number of factors of $[H_2]$
	and assume that each $v\in V$ is identified with its respective (pre-)coordinate vector
	$(v_1,\dots,v_k)$.

	In Line \ref{alg:init-Gaux} the auxiliary graph is be initialized. 
	In particular, we initialize $G_\aux$ as adjacency list, i.e., we create 
	empty lists $N[1],\dots,N[k]$ which can be done  in $O(k)$ time.
	
	We are now concerned with the for-loop in Line \ref{alg:for-start} -
	\ref{alg:for-end}. For each of the $m$ arcs we have to find the
	indices where the vertices of the particular arc differs. To this end, 
	any two vertices $x$ and $y$ of $V(e)$ are chosen and their
	$k$ coordinates are compared, which takes $O(k)$ time. The nested for-loop (Line
	\ref{alg:for2-begin} - \ref{alg:for2-end}) is executed for all coordinates $i$
	where the vertices of arc $e$ are identical and it is checked whether $\inc(e,i)=(\inc(t(e),i),
	\inc(h(e),i)$ is contained in $E_\lex$ or not. The increment $(\inc(t(e),i),
	\inc(h(e),i)$ can be computed in $O(r)$ time. Note, the vertices within
	$\inc(t(e),i)$ and $\inc(h(e),i)$ are still lexicographically ordered as only
	vertex-coordinates are incremented that have been identical for the vertices
	within the arc and thus,their $i$-th positions are
	all still equal after the computation of $\inc(e,i)$. We now check
	whether $(\inc(t(e),i), \inc(h(e),i)\notin E_\lex$. Since $E_\lex$ is already
	ordered, binary search finds the corresponding arc using at most $O(\log(m))$
	comparisons of arcs and since head and tail of each arc are in lexicographic order 
	comparing two arcs takes $O(r)$ time. 
	Therefore, the if-condition in Line \ref{alg:if} takes $O(r+r\log(m)) = 
	O(r\log(m))$ time. 
	In case,	$(\inc(t(e),i), \inc(h(e),i)\notin E_\lex$ we have to add a respective edge
	$ij$ in $G_\aux$, if not already set. Hence, to check whether $ij$ exists in
	$G_\aux$, we need to validate if $i\in N[j]$. To this end, assume that $N[j]$
	is already ordered. Hence we need $O(\log(k))$ comparisons to verify if $i\in
	N[j]$. If this is not the case $i$ is added to $N[j]$ on the respective
	position so that $N[j]$ stays sorted. Similarly, $j$ is added to $N[i]$
	whenever $i\not\in N[j]$. 
	Hence the nested for-loop in (Line \ref{alg:for2-begin} -
	\ref{alg:for2-end}) has time complexity $O(k(r\log(m)+\log(k))) =
	O(kr\log(m))$, since the number of arcs $m$ 
	is at least as big as the number (non-trivial) factors $k$. 
	Take together the latter arguments, the entire for-loop 
	in Line \ref{alg:for-start} -	\ref{alg:for-end} has time-complexity
	$O(m(k+kr\log(m))) = O(mkr\log(m))$. 
	By Lemma \ref{lem:bound} this is $O(mnr)$.
	Moreover, if we assume that the rank $r$ is bounded, then 
	$m\leq n^r$ and hence, $O(\log(m)) = O(\log(n^r)) = O(r\log(n))=O(\log(n))$
	In this case, the time complexity of Line \ref{alg:for-start} -	\ref{alg:for-end} 
	is $O(mk\log(m)) = O(mk\log(n)) =O(m\log^2(n))$. 	
	
	Finally, we use Algorithm \ref{alg:combine} which performs in $O(mnr^2)$ time and
	if the rank $r$ is bounded it has time-complexity $O(m\log^2(n))$ (Lemma \ref{lem:runtime-post}).

	To summarize, each step of Algorithm \ref{alg:pfd} can be performed in 
	$O(mnr^2)$ time and if the rank $r$ is bounded the time-complexity 
	is $O(m\log^2(n))$.
\end{proof}

\section*{Acknowledgment}
We thank the organizers of the 8th Slovenian Conference on Graph Theory (2015)
in Kranjska Gora, where the authors participated, met and basically drafted the main ideas
of this paper, while drinking a cold and tasty red Union, or was it 
a green La{\v{s}}ko?

We also thank Wilfried Imrich and Iztok Peterin for helpful comments regarding
the time-complexity of our algorithm.

\bibliographystyle{elsarticle-harv}
\bibliography{biblio}

\begin{thebibliography}{32}
\expandafter\ifx\csname natexlab\endcsname\relax\def\natexlab#1{#1}\fi
\providecommand{\url}[1]{\texttt{#1}}
\providecommand{\href}[2]{#2}
\providecommand{\path}[1]{#1}
\providecommand{\DOIprefix}{doi:}
\providecommand{\ArXivprefix}{arXiv:}
\providecommand{\URLprefix}{URL: }
\providecommand{\Pubmedprefix}{pmid:}
\providecommand{\doi}[1]{\href{http://dx.doi.org/#1}{\path{#1}}}
\providecommand{\Pubmed}[1]{\href{pmid:#1}{\path{#1}}}
\providecommand{\bibinfo}[2]{#2}
\ifx\xfnm\relax \def\xfnm[#1]{\unskip,\space#1}\fi
\bibitem[{Aurenhammer et~al.(1992)Aurenhammer, Hagauer and Imrich}]{AHI:92}
\bibinfo{author}{Aurenhammer, F.}, \bibinfo{author}{Hagauer, J.},
  \bibinfo{author}{Imrich, W.}, \bibinfo{year}{1992}.
\newblock \bibinfo{title}{Cartesian graph factorization at logarithmic cost per
  edge}.
\newblock \bibinfo{journal}{Comput. Complexity} \bibinfo{volume}{2},
  \bibinfo{pages}{331--349}.
\bibitem[{Berge(1989)}]{Berge:Hypergraphs}
\bibinfo{author}{Berge, C.}, \bibinfo{year}{1989}.
\newblock \bibinfo{title}{Hypergraphs: {C}ombinatorics of finite sets}.
  volume~\bibinfo{volume}{45}.
\newblock \bibinfo{publisher}{North-Holland}, \bibinfo{address}{Amsterdam}.
\bibitem[{Black(2015)}]{Black15}
\bibinfo{author}{Black, T.}, \bibinfo{year}{2015}.
\newblock \bibinfo{title}{Monotone properties of k-uniform hypergraphs are
  weakly evasive}, in: \bibinfo{booktitle}{Proceedings of the 2015 Conference
  on Innovations in Theoretical Computer Science}, \bibinfo{publisher}{ACM},
  \bibinfo{address}{New York, NY, USA}. pp. \bibinfo{pages}{383--391}.
\bibitem[{Bre{\v{s}}ar(2004)}]{Bresar:03}
\bibinfo{author}{Bre{\v{s}}ar, B.}, \bibinfo{year}{2004}.
\newblock \bibinfo{title}{On subgraphs of {C}artesian product graphs and
  {S}-primeness}.
\newblock \bibinfo{journal}{Discr.\ Math.} \bibinfo{volume}{282},
  \bibinfo{pages}{43--52}.
\bibitem[{Bretto(2006)}]{Bretto06:Helly}
\bibinfo{author}{Bretto, A.}, \bibinfo{year}{2006}.
\newblock \bibinfo{title}{Hypergraphs and the {H}elly property}.
\newblock \bibinfo{journal}{Ars Comb.} \bibinfo{volume}{78},
  \bibinfo{pages}{23--32}.
\bibitem[{Bretto(2013)}]{Bretto:13b}
\bibinfo{author}{Bretto, A.}, \bibinfo{year}{2013}.
\newblock \bibinfo{title}{Applications of hypergraph theory: A brief overview},
  in: \bibinfo{booktitle}{Hypergraph Theory}. \bibinfo{publisher}{Springer
  International Publishing}. Mathematical Engineering, pp.
  \bibinfo{pages}{111--116}.
\bibitem[{Bretto et~al.(2009)Bretto, Silvestre and
  Vall{\'e}e}]{Bretto09:HyperCartProd}
\bibinfo{author}{Bretto, A.}, \bibinfo{author}{Silvestre, Y.},
  \bibinfo{author}{Vall{\'e}e, T.}, \bibinfo{year}{2009}.
\newblock \bibinfo{title}{Cartesian product of hypergraphs: properties and
  algorithms}, in: \bibinfo{booktitle}{4th Athens Colloquium on Algorithms and
  Complexity (ACAC 2009)}, pp. \bibinfo{pages}{22--28}.
\bibitem[{Bretto et~al.(2013)Bretto, Silvestre and Vallée}]{Bretto:13}
\bibinfo{author}{Bretto, A.}, \bibinfo{author}{Silvestre, Y.},
  \bibinfo{author}{Vallée, T.}, \bibinfo{year}{2013}.
\newblock \bibinfo{title}{Factorization of products of hypergraphs: Structure
  and algorithms}.
\newblock \bibinfo{journal}{Theoretical Computer Science}
  \bibinfo{volume}{475}, \bibinfo{pages}{47 -- 58}.
\bibitem[{Crespelle et~al.(2013)Crespelle, Thierry and Lambert}]{CTL:13}
\bibinfo{author}{Crespelle, C.}, \bibinfo{author}{Thierry, E.},
  \bibinfo{author}{Lambert, T.}, \bibinfo{year}{2013}.
\newblock \bibinfo{title}{A linear-time algorithm for computing the prime
  decomposition of a directed graph with regard to the {C}artesian product},
  in: \bibinfo{editor}{Du, D.Z.}, \bibinfo{editor}{Zhang, G.} (Eds.),
  \bibinfo{booktitle}{Computing and Combinatorics}.
  \bibinfo{publisher}{Springer Berlin Heidelberg}. volume
  \bibinfo{volume}{7936} of \textit{\bibinfo{series}{Lecture Notes in Computer
  Science}}, pp. \bibinfo{pages}{469--480}.
\bibitem[{D\"{o}rfler(1979)}]{Doerfler79:CoversDirectProd}
\bibinfo{author}{D\"{o}rfler, W.}, \bibinfo{year}{1979}.
\newblock \bibinfo{title}{Multiple {C}overs of {H}ypergraphs}.
\newblock \bibinfo{journal}{Annals of the New York Academy of Sciences}
  \bibinfo{volume}{319}, \bibinfo{pages}{169--176}.
\bibitem[{Feder(1992)}]{Feder:92}
\bibinfo{author}{Feder, T.}, \bibinfo{year}{1992}.
\newblock \bibinfo{title}{Product graph representations}.
\newblock \bibinfo{journal}{J. Graph Theory} \bibinfo{volume}{16},
  \bibinfo{pages}{467--488}.
\bibitem[{Feigenbaum(1986)}]{Fei86}
\bibinfo{author}{Feigenbaum, J.}, \bibinfo{year}{1986}.
\newblock \bibinfo{title}{Directed {C}artesian-product graphs have unique
  factorizations that can be computed in polynomial time}.
\newblock \bibinfo{journal}{Discrete Appl. Math.} \bibinfo{volume}{15},
  \bibinfo{pages}{105 -- 110}.
\bibitem[{Feigenbaum et~al.(1985)Feigenbaum, Hershberger and
  Sch{\"a}ffer}]{FHS:85}
\bibinfo{author}{Feigenbaum, J.}, \bibinfo{author}{Hershberger, J.},
  \bibinfo{author}{Sch{\"a}ffer, A.}, \bibinfo{year}{1985}.
\newblock \bibinfo{title}{A polynomial time algorithm for finding the prime
  factors of {C}artesian-product graphs}.
\newblock \bibinfo{journal}{Discrete Appl. Math.} \bibinfo{volume}{12},
  \bibinfo{pages}{123--138}.
\bibitem[{Hammack et~al.(2011)Hammack, Imrich and
  Klav{\v{z}}ar}]{Hammack:2011a}
\bibinfo{author}{Hammack, R.}, \bibinfo{author}{Imrich, W.},
  \bibinfo{author}{Klav{\v{z}}ar, S.}, \bibinfo{year}{2011}.
\newblock \bibinfo{title}{Handbook of Product Graphs}.
\newblock Discrete Mathematics and its Applications. \bibinfo{edition}{2nd}
  ed., \bibinfo{publisher}{CRC Press}.
\bibitem[{Hellmuth(2013)}]{Hel-12}
\bibinfo{author}{Hellmuth, M.}, \bibinfo{year}{2013}.
\newblock \bibinfo{title}{On the complexity of recognizing {S}-composite and
  {S}-prime graphs}.
\newblock \bibinfo{journal}{Discrete Applied Mathematics}
  \bibinfo{volume}{161}, \bibinfo{pages}{1006 -- 1013}.
\bibitem[{Hellmuth et~al.(2014)Hellmuth, Noll and Ostermeier}]{HNO:14}
\bibinfo{author}{Hellmuth, M.}, \bibinfo{author}{Noll, M.},
  \bibinfo{author}{Ostermeier, L.}, \bibinfo{year}{2014}.
\newblock \bibinfo{title}{Strong products of hypergraphs: Unique prime
  factorization theorems and algorithms}.
\newblock \bibinfo{journal}{Discrete Applied Mathematics}
  \bibinfo{volume}{171}, \bibinfo{pages}{60 -- 71}.
\bibitem[{Hellmuth et~al.(2012a)Hellmuth, Ostermeier and Stadler}]{HGS-09}
\bibinfo{author}{Hellmuth, M.}, \bibinfo{author}{Ostermeier, L.},
  \bibinfo{author}{Stadler, P.}, \bibinfo{year}{2012}a.
\newblock \bibinfo{title}{Diagonalized {C}artesian products of {S}-prime graphs
  are {S}-prime}.
\newblock \bibinfo{journal}{Discrete Math.} \bibinfo{volume}{312},
  \bibinfo{pages}{74 -- 80}.
\newblock \bibinfo{note}{Algebraic Graph Theory - A Volume Dedicated to Gert
  Sabidussi on the Occasion of His 80th Birthday}.
\bibitem[{Hellmuth et~al.(2012b)Hellmuth, Ostermeier and Stadler}]{HOS-12}
\bibinfo{author}{Hellmuth, M.}, \bibinfo{author}{Ostermeier, L.},
  \bibinfo{author}{Stadler, P.F.}, \bibinfo{year}{2012}b.
\newblock \bibinfo{title}{A survey on hypergraph products}.
\newblock \bibinfo{journal}{Math. Comput. Sci.} \bibinfo{volume}{6},
  \bibinfo{pages}{1--32}.
\bibitem[{Imrich(1967)}]{Imrich67:Mengensysteme}
\bibinfo{author}{Imrich, W.}, \bibinfo{year}{1967}.
\newblock \bibinfo{title}{Kartesisches {P}rodukt von {M}engensystemen und
  {G}raphen}.
\newblock \bibinfo{journal}{Studia Sci. Math. Hungar.} \bibinfo{volume}{2},
  \bibinfo{pages}{285 -- 290}.
\bibitem[{Imrich(1971)}]{Imrich70:SchwachKartProd}
\bibinfo{author}{Imrich, W.}, \bibinfo{year}{1971}.
\newblock \bibinfo{title}{{\"{U}}ber das schwache {K}artesische {P}rodukt von
  {G}raphen}.
\newblock \bibinfo{journal}{Journal of Combinatorial Theory}
  \bibinfo{volume}{11}, \bibinfo{pages}{1--16}.
\bibitem[{Imrich and Peterin(2007)}]{Imrich07:linear}
\bibinfo{author}{Imrich, W.}, \bibinfo{author}{Peterin, I.},
  \bibinfo{year}{2007}.
\newblock \bibinfo{title}{Recognizing {C}artesian products in linear time}.
\newblock \bibinfo{journal}{Discrete Math.} \bibinfo{volume}{307},
  \bibinfo{pages}{472--483}.
\bibitem[{Kaveh and Alinejad(2012)}]{KA:12}
\bibinfo{author}{Kaveh, A.}, \bibinfo{author}{Alinejad, B.},
  \bibinfo{year}{2012}.
\newblock \bibinfo{title}{Hypergraph products for structural mechanics}, in:
  \bibinfo{editor}{Topping, B.} (Ed.), \bibinfo{booktitle}{Proceedings of the
  Eleventh International Conference on Computational Structures Technology},
  \bibinfo{publisher}{Civil-Comp Press}, \bibinfo{address}{Stirlingshire, UK}.
\newblock \DOIprefix\doi{10.4203/ccp.99.266}. \bibinfo{note}{paper 266}.
\bibitem[{Kaveh and Alinejad(2015)}]{KA:15}
\bibinfo{author}{Kaveh, A.}, \bibinfo{author}{Alinejad, B.},
  \bibinfo{year}{2015}.
\newblock \bibinfo{title}{Hypergraph products for structural mechanics}.
\newblock \bibinfo{journal}{Advances in Engineering Software}
  \bibinfo{volume}{80}, \bibinfo{pages}{72 -- 81}.
\newblock \bibinfo{note}{Civil-Comp}.
\bibitem[{Klav{\v{z}}ar et~al.(2002)Klav{\v{z}}ar, Lipovec and
  Petkov{\v{s}}ek}]{Klavzar:02}
\bibinfo{author}{Klav{\v{z}}ar, S.}, \bibinfo{author}{Lipovec, A.},
  \bibinfo{author}{Petkov{\v{s}}ek, M.}, \bibinfo{year}{2002}.
\newblock \bibinfo{title}{On subgraphs of {C}artesian product graphs}.
\newblock \bibinfo{journal}{Discr.\ Math.} \bibinfo{volume}{244},
  \bibinfo{pages}{223--230}.
\bibitem[{Ostermeier et~al.(2012)Ostermeier, Hellmuth and
  Stadler}]{OstHellmStad11:CartProd}
\bibinfo{author}{Ostermeier, L.}, \bibinfo{author}{Hellmuth, M.},
  \bibinfo{author}{Stadler, P.F.}, \bibinfo{year}{2012}.
\newblock \bibinfo{title}{The {C}artesian product of hypergraphs}.
\newblock \bibinfo{journal}{Journal of Graph Theory} \bibinfo{volume}{70},
  \bibinfo{pages}{180--196}.
\bibitem[{Sabidussi(1960)}]{Sabidussi:60}
\bibinfo{author}{Sabidussi, G.}, \bibinfo{year}{1960}.
\newblock \bibinfo{title}{Graph {M}ultiplication}.
\newblock \bibinfo{journal}{Mathematische Zeitschrift} \bibinfo{volume}{72},
  \bibinfo{pages}{446--457}.
\bibitem[{Shapiro(1953)}]{Shapiro:53}
\bibinfo{author}{Shapiro, H.}, \bibinfo{year}{1953}.
\newblock \bibinfo{title}{The embedding of graphs in cubes and the design of
  sequential relay circuits}.
\newblock \bibinfo{journal}{Bell Telephone Laboratories Memorandum}
  \bibinfo{note}{Unpublished}.
\bibitem[{Sonntag(1989)}]{Sonntag89:HamCart}
\bibinfo{author}{Sonntag, M.}, \bibinfo{year}{1989}.
\newblock \bibinfo{title}{{Hamiltonian properties of the {C}artesian sum of
  hypergraphs.}}
\newblock \bibinfo{journal}{J. Inf. Process. Cybern.} \bibinfo{volume}{25},
  \bibinfo{pages}{87--100}.
\bibitem[{Szamko{\l}owisz(1962)}]{Szamko:62}
\bibinfo{author}{Szamko{\l}owisz, L.}, \bibinfo{year}{1962}.
\newblock \bibinfo{title}{Remarks on the {C}artesian product of two graphs}.
\newblock \bibinfo{journal}{Colloq. Math.} \bibinfo{volume}{9},
  \bibinfo{pages}{43--47}.
\bibitem[{Vizing(1963)}]{Vizing63:CartProd}
\bibinfo{author}{Vizing, V.G.}, \bibinfo{year}{1963}.
\newblock \bibinfo{title}{The {C}artesian product of graphs}.
\newblock \bibinfo{journal}{Vy\v cisl. Sistemy No.} \bibinfo{volume}{9},
  \bibinfo{pages}{30--43}.
\bibitem[{Winkler(1987)}]{Winkler:87}
\bibinfo{author}{Winkler, P.}, \bibinfo{year}{1987}.
\newblock \bibinfo{title}{Factoring a graph in polynomial time}.
\newblock \bibinfo{journal}{European J. Combin.} \bibinfo{volume}{8},
  \bibinfo{pages}{209--212}.
\bibitem[{Zhu(1992)}]{Zhu:92}
\bibinfo{author}{Zhu, X.}, \bibinfo{year}{1992}.
\newblock \bibinfo{title}{On the chromatic number of the product of
  hypergraphs}.
\newblock \bibinfo{journal}{Ars Comb.} \bibinfo{volume}{34},
  \bibinfo{pages}{25--31}.

\end{thebibliography}

\end{document}